\documentclass{amsart}
\usepackage[a4paper, total={6in, 8in}]{geometry}
\usepackage{enumerate}
\usepackage{amsmath}
\usepackage{amssymb,latexsym}
\usepackage{amsthm}
\usepackage{color}
\usepackage{cancel}
\usepackage{graphicx}
\usepackage{cite}
\usepackage{longtable}
\makeatletter
\usepackage{comment}
\let\wfs@comment@comment\comment
\let\comment\@undefined

\usepackage{changes}
\let\wfs@changes@comment\comment
\let\comment\@undefined

\newcommand\comment{%
    \ifthenelse{\equal{\@currenvir}{comment}}
    {\wfs@comment@comment}
    {\wfs@changes@comment}%
}

\definechangesauthor[name=Daniele, color=red]{DAN}
\definechangesauthor[name=Giovanni, color=blue]{GIO}
\definechangesauthor[name=Ferd, color=green]{FE}
\usepackage{lineno}
\usepackage{hyperref}
\usepackage{comment}
\usepackage{amscd}

\DeclareMathOperator{\C}{\mathcal{C}}

\newtheorem*{main}{Main Theorem}

\newtheorem{theorem}{Theorem}[section]

\newtheorem{lemma}[theorem]{Lemma}
\newtheorem{corollary}[theorem]{Corollary}
\newtheorem{definition}[theorem]{Definition}
\newtheorem{proposition}[theorem]{Proposition}

\newtheorem{example}[theorem]{Example}

\newtheorem{remark}[theorem]{Remark}

\newtheorem{ass}[theorem]{Assumptions}

\newcommand{\fqn}{\mathbb{F}_{q^n}}

\newcommand{\cC}{{\mathcal C}}

\newcommand{\F}{{\mathbb F}}

\newcommand{\K}{{\mathbb K}}
\newcommand{\fq}{{\mathbb F}_{q}}

\title{Linear maximum rank distance codes of exceptional type}
\begin{small}
\author[D. Bartoli]{Daniele Bartoli}
\address{Daniele Bartoli, \textnormal{Dipartimento di Matematica e Informatica, Universit\`a degli Studi di Perugia,  Perugia, Italy}}
\email{daniele.bartoli@unipg.it}

\author[G. Zini]{Giovanni Zini}
\author[F. Zullo]{Ferdinando Zullo}
\address{Giovanni Zini, Ferdinando Zullo, \textnormal{Dipartimento di Matematica e Fisica, Universit\`a degli Studi della Campania ``Luigi Vanvitelli'', Viale Lincoln, 5, I--\,81100 Caserta, Italy}}
\email{\{giovanni.zini,ferdinando.zullo\}@unicampania.it}
\end{small}


\begin{document}

\maketitle





\begin{abstract}
Scattered polynomials of a given index over finite fields are intriguing rare objects with many connections within mathematics. Of particular interest are the \emph{exceptional} ones, as defined in 2018 by the first author and Zhou, for which partial classification results are known.
In this paper we propose a unified  algebraic description of $\mathbb{F}_{q^n}$-linear maximum rank distance codes, introducing  the notion of \emph{exceptional} linear maximum rank distance codes of a given index. Such a connection  naturally extends the notion of exceptionality for a scattered polynomial in the rank metric framework and provides a generalization of Moore sets in the monomial MRD context.
We move towards the classification of exceptional linear MRD codes, by showing that the ones of index zero are generalized Gabidulin codes and proving that in the positive index case the code contains an exceptional scattered polynomial of the same index.
\end{abstract}

\section{Introduction}
Let $q$ be a prime power, $n$ be a positive integer, and denote by $\mathbb{F}_{q^n}$ the finite field with $q^n$ elements. 
Rank metric codes can be seen as sets of $\fq$-linear endomorphisms of $\fqn$ equipped with the rank distance, that is the distance between two elements is defined as the (linear algebraic) rank of their difference.
Since the $\fq$-algebra of the $\fq$-linear endomorphisms of $\fqn$ and the $\fq$-algebra $\mathcal{L}_{n,q}=\left\{ \sum_{i=0}^{n-1} a_ix^{q^i} \colon a_i \in \fqn \right\}$ of $\F_q$-linearized polynomials  of $q$-degree smaller than $n$ are isomorphic, each rank metric code can be also seen as a subset of $\mathcal{L}_{n,q}$.
For any $f(x)=\sum_{i=0}^{n-1}a_ix^{q^i}$, we define $\deg_q(f(x))=\max\{i \colon a_i\ne 0\}$ and $\mathrm{mindeg}_q(f(x))=\min\{i \colon a_i\ne 0\}$.
In this context, a rank metric code $\mathcal{C}\subseteq \mathcal{L}_{n,q}$ with minimum distance $d$ achieving the equality in the Singleton-like bound
\begin{equation}\label{eq:Singleton}
    |\C| \leq q^{n(n-d+1)}
\end{equation}
is  called \emph{maximum rank distance} (MRD) code.
Rank metric codes and in particular MRD codes have been introduced several times \cite{delsarte1978bilinear,gabidulin1985theory} and have been widely investigated in the last few years, due to applications in network coding \cite{MR2450762} and cryptography \cite{MR3678916}.
Two rank metric codes $\C_1, \C_2\subseteq\mathcal{L}_{n,q}$ are said to be \emph{equivalent} if there exist two invertible $\fq$-linearized polynomials $g(x),h(x) \in \mathcal{L}_{n,q}$ and a field automorphism $\rho \in \mathrm{Aut}(\fqn)$ such that
\[ \C_1=g \circ \C_2^\rho \circ h= \{ g \circ f^\rho \circ h \colon f \in \C_2\}, \]
where $f^\rho(x):=\sum_{i=0}^{n-1}\rho(a_i)x^{\sigma^i}$ if $f(x)=\sum_{i=0}^{n-1}a_ix^{\sigma^i}$ and the composition has to be considered modulo $x^{q^n}-x$.
In order to study the equivalence between two rank metric codes, one can make use of the idealisers. They have been introduced in \cite{liebhold2016automorphism}, where the \emph{left} and \emph{right idealisers} of a rank metric code
$\C\subseteq\mathcal{L}_{n,q}$ are defined respectively as
\[L(\C)=\{\varphi(x) \in \mathcal{L}_{n,q} \colon \varphi \circ f \in \C\, \text{for all} \, f \in \C\},\quad R(\C)=\{\varphi(x) \in \mathcal{L}_{n,q} \colon f \circ \varphi \in \C\, \text{for all} \, f \in \C\}.\]
Such objects have also been investigated in \cite{lunardon2018nuclei}, where they have been called respectively middle and right nuclei.

In this paper we are interested in \emph{$\fqn$-linear} MRD codes, that is MRD codes $\C$ such that $L(\C)$ is equivalent to $\mathcal{F}_n=\{\alpha x \colon \alpha \in \fqn\}$; see \cite[Definition 12]{sheekeysurvey}.
Thus, every $\fqn$-linear MRD code is equivalent to an $\fqn$-subspace of $\mathcal{L}_{n,q}$ and we will always consider MRD codes which are $\fqn$-subspaces of $\mathcal{L}_{n,q}$.
The MRD condition for $\fqn$-linear rank metric codes reads as follows.
By \eqref{eq:Singleton}, an $\fqn$-linear rank metric code $\C\subseteq\mathcal{L}_{n,q}$,
with $\dim_{\fqn}(\C)=k$ and minimum distance $d$, is an MRD code  if and only if $d=n-k+1$, or equivalently,
\[ \dim_{\fq}(\ker(f))\leq k-1,\,\,\,\text{for all}\,\,\, f \in \C\setminus\{0\}. \]

This paper is devoted to the investigation of $\fqn$-linear MRD codes which are \emph{exceptional}.
An $\fqn$-linear MRD code $\C\subseteq\mathcal{L}_{n,q}$ is an \emph{exceptional MRD code} if the rank metric code
\[ \C_m=\langle\mathcal{C}\rangle_{\mathbb{F}_{q^{nm}}}\subseteq \mathcal{L}_{nm,q} \]
is an MRD code for infinitely many $m$.
Only two families of exceptional $\fqn$-linear MRD codes are known:
\begin{itemize}
    \item[(G)] $\mathcal{G}_{k,s}=\langle x,x^{q^s},\ldots,x^{q^{s(k-1)}}\rangle_{\fqn}$, with $\gcd(s,n)=1$, see \cite{delsarte1978bilinear,gabidulin1985theory,kshevetskiy2005new};
    \item[(T)] $\mathcal{H}_{k,s}(\delta)=\langle x^{q^s},\ldots,x^{q^{s(k-1)}},x+\delta x^{q^{sk}}\rangle_{\fqn}$, with $\gcd(s,n)=1$ and $\mathrm{N}_{q^n/q}(\delta)\neq (-1)^{nk}$, see  \cite{MR3543528,lunardon2018generalized}. 
\end{itemize}
The first family is known as \emph{generalized Gabidulin codes} and the second one as \emph{generalized twisted Gabidulin codes}.

Although the definition of exceptional $\fqn$-linear MRD codes appears in this paper for the first time, it has been already studied in particular subcases in different contexts.

In the case $k=2$, exceptional MRD codes have been considered via  so-called exceptional scattered polynomials. Let $f(x) \in \mathcal{L}_{n,q}$ and $t$ be a nonnegative integer $t\leq n-1$. Then $f$ is said to be \emph{scattered of index $t$} if for every $x,y \in \fqn^*$
\[ \frac{f(x)}{x^{q^t}}=\frac{f(y)}{y^{q^t}}\,\, \Leftrightarrow\,\, \frac{y}x\in \fq, \]
or equivalently,
\begin{equation}\label{eq:condscatt} \dim_{\fq}(\ker(f(x)-mx^{q^t}))\leq 1, \,\,\,\text{for every}\,\,\, m \in \fqn. \end{equation}
The term \emph{scattered} arises from a geometric framework; see \cite{blokhuis2000scattered}. 
Indeed, $f$ is scattered of index $t$ if and only the $\fq$-subspace 
\[ U_{t,f}=\{ (x^{q^t},f(x)) \colon x \in \fqn \} \]
has the property that $\dim_{\fq}(U_{t,f}\cap \langle \mathbf{v}\rangle_{\fqn})\leq 1$ for every nonzero vector $\mathbf{v}\in \fqn^2$, that is $U_{t,f}$ is scattered with respect to the Desarguesian spread $\{\langle \mathbf{v}\rangle_{\fqn} \colon \mathbf{v}\in \fqn^2\setminus\{(0,0)\}\}$.
Sheekey in \cite{MR3543528}, taking into account \eqref{eq:condscatt}, pointed out the following  connection between scattered polynomials and $\fqn$-linear MRD codes:  $f$ is scattered of index $t$ if and only if $\C_{f,t}=\langle x^{q^t}, f(x) \rangle_{\fqn}$ is an MRD code with $\dim_{\fqn}(\C)=2$.
The polynomial $f$ is said to be \emph{exceptional scattered} of index $t$ if it is scattered of index $t$ as a polynomial in $\mathcal{L}_{nm,q}$, for infinitely many $m$; see \cite{MR3812212}.
Taking into account \eqref{eq:condscatt}, a polynomial $f$ is exceptional scattered of index $t$ if and only if the corresponding MRD code $\C_{f,t}$ is exceptional. 
While several families of scattered polynomials have been constructed in recent years \cite{MR3543528,lunardon2018generalized,lunardon2000blocking,zanella2019condition,MR4173668,longobardi2021linear,longobardi2021large,NPZ,zanella2020vertex,csajbok2018new,csajbok2018new2,marino2020mrd,blokhuis2000scattered}, only two families of exceptional ones are known:
\begin{itemize}
    \item[(Ps)] $f(x)=x^{q^s}$ of index $0$, with $\gcd(s,n)=1$ (polynomials of so-called pseudoregulus type);
    \item[(LP)] $f(x)=x+\delta x^{q^{2s}}$ of index $s$, with $\gcd(s,n)=1$ and $\mathrm{N}_{q^n/q}(\delta)\ne1$ (so-called LP polynomials).
\end{itemize}
Such two families correspond to the known exceptional $\mathbb{F}_{q^n}$-linear MRD  codes (G) and (T).

Several tools have already been proposed in the study of exceptional scattered polynomials, related to  algebraic curves or Galois extensions of function fields; see \cite{Bartoli:2020aa4,MR3812212,MR4190573,MR4163074,Bartoli:2021}.
However, their classification is still unknown when the index is greater than $1$.

For $k>2$, the only  known families of $\fqn$-subspaces of $\mathcal{L}_{q,n}$ corresponding to MRD codes are (G) and (T) as described above and Delsarte dual codes of the MRD codes associated with scattered polynomials.
In \cite{MR4110235} it has been shown that the only exceptional $\fqn$-linear MRD codes spanned by monomials are the codes (G), in connection with so-called \emph{Moore exponent sets}.

It is therefore natural to investigate exceptional $\fqn$-linear MRD codes not generated by monomials.
To this aim, we generalize the notion of Moore exponent set; see Section \ref{sec:MooreMRD}.


Using a connection between the generators of  $\mathbb{F}_{q^n}$-linear rank metric codes $\mathcal{C}$  and certain algebraic hypersurfaces $\mathcal{X}_{\mathcal{C}}$, we obtain a partial classification of exceptional $\mathbb{F}_{q^n}$-linear MRD codes. Tools from intersection theory (see Section \ref{sec:prelimvar})  yield sufficient conditions on the generators for $\mathcal{C}$ to be MRD, via the existence of $\mathbb{F}_{q^n}$-rational absolutely irreducible components  of $\mathcal{X}_{\mathcal{C}}$.


Our main results can be summarized as follows.
\begin{main}
Let $\mathcal{C}\subseteq\mathcal{L}_{n,q}$ be an exceptional $k$-dimensional $\fqn$-linear MRD code containing at least a separable polynomial $f(x)$ and a monomial.
If $k>3$, assume also that $q>5$.
Let $t$ be the minimum integer such that $x^{q^t}\in \mathcal{C}$.
\begin{itemize}
    \item If $t=0$ and $\mathcal{C}=\langle x^{q^t},g_2(x),g_3(x),\ldots,g_k(x)\rangle_{\mathbb{F}_{q^n}}$, with $\deg_{q}(g_2(x))<\cdots<\deg_q(g_k(x))$ and $(q,\deg_q(g_2(x)))\notin\{(2,2),(2,4),(3,2),(4,2),(5,2)\}$, then $\C$ is a generalized Gabidulin code.
    \item If $t>0$ and $\mathcal{C}=\langle x^{q^t},f(x),g_3(x),\ldots,g_k(x)\rangle_{\mathbb{F}_{q^n}}$, with $\deg(g_i(x))>\max\{q^t,\deg(f(x))\}$ for each $i=3,\ldots,k$, then $f(x)$ is exceptional scattered of index $t$.
\end{itemize}
\end{main}

When $\mathcal{C}$ contains a separable polynomial and a monomial, we call the non-negative integer $t$ of Main Theorem the \emph{index} of $\mathcal{C}$.

\section{Preliminaries on algebraic varieties}\label{sec:prelimvar}
An algebraic hypersurface is an algebraic variety that can be defined by a single polynomial  equation. An algebraic hypersurface defined over a field $\K$  is \emph{absolutely irreducible}  if the associated polynomial is irreducible over every algebraic extension of $\K$. An absolutely irreducible $\mathbb{K}$-rational component of a hypersurface $\mathcal{V}$, defined by the polynomial $F$, is simply an absolutely irreducible hypersurface which is associated to a factor of $F$ defined over $\mathbb{K}$. For a finite field $\mathbb{F}_q$, let $\overline{\mathbb{F}}_q$ denote its algebraic closure. Also, $\mathbb{P}^m(\mathbb{K})$ (resp. $\mathbb{A}^m(\mathbb{K})$) denotes the $m$-dimensional projective (resp. affine) space over the field $\mathbb{K}$.

We recall some known results on algebraic hypersurfaces of which our approach will make use. 


\begin{lemma}\label{le:subvarieties}\cite[Lemma 2.1]{MR2648536}
Let $\mathcal{H}$ be an absolutely irreducible hypersurface and $\mathcal{X}$ be an $\mathbb{F}_q$-rational hypersurface of $\mathbb{P}^m(\overline{\mathbb{F}}_{q})$. If $\mathcal{X}\cap \mathcal{H}$ has a non-repeated $\mathbb{F}_q$-rational absolutely irreducible component, then $\mathcal{X}$ has a non-repeated $\mathbb{F}_q$-rational absolutely irreducible component.
\end{lemma}

With the symbol $I(P, \mathcal{A}\cap \mathcal{B})$ we denote the intersection multiplicity of two plane curves in $\mathbb{A}^2(\mathbb{K})$ at a point $P\in \mathbb{A}^2(\mathbb{K})$. Classical results on such an integer can be found in most of the textbooks on algebraic curves. For other concepts related to algebraic varieties we refer to \cite{MR0463157}. For the special case of curves, a good reference is \cite{MR1042981}.



\begin{lemma}\cite[Proposition 2]{MR1359909}
	\label{le:intersection_number_m_m1_coprime}
	Let $F\in\F_q[X,Y]$ be such that $F=AB$ for some $A,B\in\overline{\mathbb{F}}_q[X,Y]$. Let $P=(u,v)$ be a point in the affine plane $\mathbb{A}^2(\overline{\mathbb{F}}_q)$ and write
	\[
	F(X+u,Y+v)=F_m(X,Y)+F_{m+1}(X,Y)+\cdots,
	\]
	where $F_i$ is zero or homogeneous of degree $i$ and $F_m\ne 0$.
	Suppose that $F_m=L^m$ for some linear polynomial $L\in\overline{\mathbb{F}}_q[X,Y]$ such that $L\nmid F_{m+1}$. 
	Then $I(P, \mathcal{A}\cap \mathcal{B})=0$, where $\mathcal{A}$ and $\mathcal{B}$ are the curves defined by $A$ and $B$ respectively.
\end{lemma}

\begin{lemma}\cite[Lemma 4.3]{MR3239294}, \cite[Lemma 2.5]{MR4110235}
	\label{le:intersection_number_linear_term}
	Let $F\in\F_q[X,Y]$ be such that $F=AB$ for some $A,B\in\overline{\mathbb{F}}_q[X,Y]$. Let $P=(u,v)$ be a point in the affine plane $\mathbb{A}^2(\overline{\mathbb{F}}_q)$ and write
	\[
	F(X+u,Y+v)=F_m(X,Y)+F_{m+1}(X,Y)+\cdots,
	\]
	where $F_i$ is zero or homogeneous of degree $i$ and $F_m\ne 0$.
	Suppose that $F_m=L^m$ for some linear polynomial $L\in\overline{\mathbb{F}}_q[X,Y]$ such that $L\mid F_{m+1}$ and $L^2\nmid F_{m+1}$. 
	Then $I(P, \mathcal{A}\cap \mathcal{B})=0$ or $m$, where $\mathcal{A}$ and $\mathcal{B}$ are the curves defined by $A$ and $B$ respectively.
\end{lemma}

\begin{lemma}\label{le:intersection_number_ordinary}\cite[Section 3.3]{MR1042981}
Let $F\in\F_q[X,Y]$ be such that $F=AB$ for some $A,B\in\overline{\mathbb{F}}_q[X,Y]$.
Let $P=(u,v)$ be a point in the affine plane $\mathbb{A}^2(\overline{\mathbb{F}}_q)$ and write
	\[
	F(X+u,Y+v)=F_m(X,Y)+F_{m+1}(X,Y)+\cdots,
	\]
	where $F_i$ is zero or homogeneous of degree $i$ and $F_m\ne 0$.
	Suppose that $F_m$ factors into $m$ distinct linear factors in  $\overline{\mathbb{F}}_q[X,Y]$. Then $I(P, \mathcal{A}\cap \mathcal{B})\leq \lfloor m^2/2\rfloor$.
\end{lemma}

\begin{lemma}\label{lemma:bello}
Let $F(X_1,\ldots,X_n) \in \mathbb{F}_{q}[X_1,\ldots,X_n]$ with
\[
	F(X_1,\ldots,X_n)=F_m(X_1,\ldots,X_n)+F_{m+1}(X_1,\ldots,X_n) +\cdots+F_M(X_1,\ldots,X_n),
\]
where $F_i$ is zero or homogeneous of degree $i$ and $F_m F_M\ne 0$. If $F_m$ or $F_M$ contains a non-repeated absolutely irreducible $\mathbb{F}_q$-rational factor, then $F(X_1,\ldots,X_n)$ contains a non-repeated absolutely irreducible $\mathbb{F}_q$-rational factor.
\end{lemma}
\begin{proof}
Let $G$ be the non-repeated absolutely irreducible $\mathbb{F}_q$-rational factor in $F_m$ (resp. $F_M$). Consider the unique absolutely irreducible factor $F^{\prime}$ of $F$ such that $G\mid F^{\prime}_m$ (resp. $G\mid F^{\prime}_M$). If $F^\prime$ were not defined over $\F_q$, then there would exist another absolutely irreducible factor $F^{\prime\prime}=\sigma(F^{\prime})\neq F^{\prime}$ of $F$ satisfying $G\mid F^{\prime\prime}_m$ (resp. $G\mid F^{\prime\prime}_M$), where $\sigma$ is the $q$-Frobenius automorphism of $\overline{\mathbb{F}}_q[X_1,\ldots,X_n]$, whence  $G^2\mid F_m$  (resp. $G^2\mid F_M$), a contradiction.
\end{proof}
In the sequel we will investigate hypersurfaces connected with Moore polynomial sets; see Definition \ref{def:moore}. In particular, we are interested in getting information on the existence of $\mathbb{F}_q$-rational absolutely irreducible components of curves contained in such hypersurfaces.

The approach that we follow has been used for the first time by Janwa, McGuire and Wilson \cite{MR1359909} to classify functions on $\mathbb{F}_{p^n}$ that are almost perfect nonlinear for infinitely many $n$, in particular for monomial functions. It can be summarized by the following theorem.

\begin{theorem}\cite[Lemma 2]{MR3326175}\label{cor:lowergrado}
Let $\mathcal{C}\subset\mathbb{P}^2(\mathbb{F}_q)$ be a curve of degree $d$ and let  $\mathcal{S}$ be the set of its singular points. Also, let $i(P)$ denote the maximum possible intersection multiplicity of two putative components of $\mathcal{C}$ at $P\in \mathcal{C}$. If 
$$\sum_{P\in \mathcal{S}} i(P) < \frac{2d^2}{9},$$
then $\mathcal{C}$ possesses at least one absolutely irreducible component defined over $\mathbb{F}_q$.
\end{theorem}

\section{Moore polynomial sets and MRD codes}\label{sec:MooreMRD}

Let $q$ be a prime power and $n$ be a positive integer. 
Consider $k$ $\fqn$-linearly independent polynomials $f_1(x), f_2(x),\ldots,f_{k}(x)\in\mathcal{L}_{n,q}$ and denote by $\underline{f}$ the $k$-tuple $(f_1(x),\ldots,f_{k}(x))$. Define
\[ M_{\underline{f}}(x_1,\ldots,x_k)=\left( 
\begin{array}{cccc}
f_1(x_1) & f_2(x_1) & \cdots & f_{k}(x_1)  \\
f_1(x_2) & f_2(x_2) & \cdots & f_{k}(x_2)  \\
\vdots & \vdots & \ddots & \vdots  \\
f_1(x_k) & f_2(x_k) & \cdots & f_{k}(x_k)  \\
\end{array}
\right). \]
For any $A=\{\alpha_1,\ldots,\alpha_{k}\}\subseteq \fqn$, define $M_{\underline{f},A}=M_{\underline{f}}(\alpha_1,\ldots,\alpha_k)$.
\begin{lemma}\label{lemma:direzioneovvia}
If $\alpha_1,\ldots,\alpha_{k}$ are $\fq$-linearly dependent, then $\det(M_{\underline{f},A})=0$.
\end{lemma}
\begin{proof}
Without loss of generality, suppose that $k\geq2$ and $\alpha_1=\sum_{i=2}^{k} b_i\alpha_i$ with $a_i \in \fq$.
Then
\[ M_{\underline{f},A}=\left( 
\begin{array}{cccc}
\sum_{i=2}^{k}b_if_1(\alpha_i) & \sum_{i=2}^{k}b_if_2(\alpha_i) & \cdots & \sum_{i=2}^{k}b_if_{k}(\alpha_i)   \\
f_1(\alpha_2) & f_2(\alpha_2) & \cdots & f_{k}(\alpha_2)  \\
\vdots & \vdots & \ddots & \vdots  \\
f_1(\alpha_{k}) & f_2(\alpha_{k}) & \cdots & f_{k}(\alpha_{k})  \\
\end{array}
\right), \]
so that the first row of $M_{\underline{f},A}$ is a linear combination of the remaining rows. Then $\det(M_{\underline{f},A})=0$.
\end{proof}

The converse of Lemma \ref{lemma:direzioneovvia} is not true in general, the following being a counterexample.

\begin{example}
Let $k$ and $n$ be positive integers with $n$ even and $k\leq n$. Consider $\underline{f}=(x,x^{q^2},\ldots,x^{q^{2(k-1)}})$.
Let $A=\{\alpha_1,\ldots,\alpha_k\}$ be a subset of $\F_{q^n}$. Then $M_{\underline{f},A}$ is the Moore matrix 
\[ M_{\underline{f},A}=\left( 
\begin{array}{cccc}
\alpha_1 &  \alpha_1^{q^2} & \cdots & \alpha_1^{q^{2(k-1)}}\\
\alpha_2 &  \alpha_2^{q^2} & \cdots & \alpha_2^{q^{2(k-1)}}\\
\vdots & \vdots & \ddots & \vdots \\
\alpha_k &  \alpha_k^{q^2} & \cdots & \alpha_k^{q^{2(k-1)}}
\end{array}\right), \]
and $\det(M_{\underline{f},A})=0$ if and only if the elements $\alpha_1,\ldots,\alpha_k$ are $\F_{q^2}$-linearly independent; see \cite[Corollary 2.1.95]{MR3087321}.
Therefore, if $\alpha_1,\ldots,\alpha_{k-1}$ are $\F_q$-linearly independent elements in $\F_{q^n}$ and $\alpha_k\in \langle \alpha_1,\ldots,\alpha_{k-1} \rangle_{\F_{q^2}}\setminus \langle \alpha_1,\ldots,\alpha_{k-1} \rangle_{\F_{q}}$, then $\det(M_{\underline{f},A})=0$ even though $\{ \alpha_1,\ldots,\alpha_{k} \}$ are $\F_q$-linearly independent.
\end{example}

The following definition identifies the tuples $\underline{f}$ for which the converse of Lemma \ref{lemma:direzioneovvia} holds and it will be crucial in our investigation for exceptional MRD codes.

\begin{definition}\label{def:moore}
Let $\underline{f}=(f_1(x),\ldots,f_k(x))$, where $k$ is a positive integer and $f_1(x),\ldots,f_k(x)\in\mathcal{L}_{n,q}$.
We say that $\underline{f}$ is a \emph{Moore polynomial set} for $q$ and $n$ if, for any $\alpha_1,\ldots,\alpha_{k}\in \fqn$,
\[
\det \left( 
\begin{array}{cccc}
f_1(\alpha_1) & f_2(\alpha_1) & \cdots & f_{k}(\alpha_1)  \\
f_1(\alpha_2) & f_2(\alpha_2) & \cdots & f_{k}(\alpha_2)  \\
\vdots & \vdots & \ddots & \vdots  \\
f_1(\alpha_{k}) & f_2(\alpha_{k}) & \cdots & f_{k}(\alpha_{k})  \\
\end{array}
\right) = 0 \quad\Longrightarrow\quad \dim_{\fq}\langle\alpha_1,\ldots,\alpha_k\rangle_{\fq}<k.
\]
If $\underline{f}$ is a Moore polynomial set for $q$ and $nm$ for infinitely many $m$, we say that $\underline{f}$ is an \emph{exceptional} Moore polynomial set for $q$ and $n$.
\end{definition}

Moore polynomial sets can be characterized in terms of MRD codes as follows.

\begin{theorem}
Let $k$ and $n$ be positive integers with $k\leq n$, and denote by  $\underline{f}$ the $k$-tuple $(f_1(x),\ldots, f_{k}(x))$, where $f_1(k),\ldots,f_k(x)\in\mathcal{L}_{n,q}$ are $\fqn$-linearly independent.
The $\fqn$-linear rank metric code
\[ \C_{\underline{f}}=\langle f_1(x),\ldots, f_{k}(x) \rangle_{\fqn} \]
is an MRD code if and only if $\underline{f}$ is a Moore polynomial set for $q$ and $n$.
\end{theorem}
\begin{proof}
Suppose that $M_{\underline{f},A}$ is singular for some $A=\{\alpha_1,\ldots,\alpha_{k}\}\subseteq \fqn$, that is, there exist $b_1,\ldots,b_{k}\in \fqn$ such that $\sum_{i=1}^{k} b_i f_i(\alpha_j)=0$, for every $j \in \{1,\ldots,k\}$. 
This means that $A$ is contained in the kernel of $F(x)=\sum_{i=1}^{k} b_i f_i(x) \in \C_{\underline{f}}$.
Since $\C_{\underline{f}}$ is an MRD code, it follows that $\dim_{\fq}(\ker(F))\leq k-1$, and hence $\alpha_1,\ldots,\alpha_{k}$ are $\fq$-linearly dependent.

Conversely, suppose that $\underline{f}$ is a Moore polynomial set for $q$ and $n$.
Assume by contradiction that there exists $g(x)\in\C_{\underline{f}}$ with $\dim_{\fq}(\ker(g(x)))\geq k$ and write $g(x)=\sum_{i=1}^k b_i f_i(x)$ with $b_i\in\fqn$.
Let $A=\{\alpha_1,\ldots,\alpha_k\}\subseteq\ker(g(x))$ where $\alpha_1,\ldots,\alpha_k$ are $\fq$-linearly independent.
Then $M_{\underline{f},A}$ is singular because its columns are $\fqn$-linearly dependent through $\sum_{i=1}^k b_i f_i(\alpha_j)=0$ for all $j=1,\ldots,k$.
Therefore, $\underline{f}$ is not a Moore polynomial set for $q$ and $n$.
\end{proof}


As a natural consequence, a characterization of the exceptionality property is obtained.

\begin{corollary}\label{cor:equiv}
Let $\C\subseteq\mathcal{L}_{n,q}$ be an $\fqn$-linear rank metric code. 
The following are equivalent:
\begin{itemize}
    \item $\C$ is an exceptional MRD code.
    \item Every $\fqn$-basis $\{f_1(x),\ldots,f_k(x)\}$ of $\C$ defines an exceptional Moore polynomial set $\underline{f}=(f_1(x),\ldots,f_k(x))$ for $q$ and $n$.
    \item There exists an $\fqn$-basis $\{f_1(x),\ldots,f_k(x)\}$ of $\C$ for which $\underline{f}=(f_1(x),\ldots,f_k(x))$ is an exceptional Moore polynomial set for $q$ and $n$.
\end{itemize}
\end{corollary}

We will investigate exceptional MRD codes by means of exceptional Moore polynomial sets.

\section{Moore polynomial sets and  varieties over finite fields}

In this section we study exceptional $\fqn$-linear MRD codes $\C\subseteq\mathcal{L}_{n,q}$ of dimension $k$ under the assumption that $\C$ contains a monomial.
Up to equivalence, we can assume that $\C$ contains a separable polynomial.
We denote by $t$ the smallest non-negative integer such that $x^{q^t}\in\C$.

\begin{remark}
If $\C$ is an $\fqn$-linear MRD code in $\mathcal{L}_{n,q}$, then $\cC$ contains an invertible map $f(x)$ (see \cite[Lemma 2.1]{lunardon2018nuclei} and\cite[Lemma 52]{ravagnani2016rank}), and hence $f^{-1}\circ\cC$ contains the identity $x$.
If $\C$ is exceptional, then $\max\{\deg_q(g(x))\colon g(x)\in \langle\C\rangle_{\mathbb{F}_{q^{nm}}} \}$ does not depend on the infinitely many $m$'s for which $\langle\C\rangle_{\mathbb{F}_{q^{nm}}}$ is MRD.
On the contrary, $\max\{\deg_q(g(x))\colon g(x)\in \langle f^{-1}\circ\C\rangle_{\mathbb{F}_{q^{nm}}} \}$ may depend on $m$, so that $f^{-1}\circ\C$ may not be exceptional.

On the other hand, the assumption that $\C$ contains a separable polynomial does not affect the exceptionality of $\C$, since $\max\{\deg_q(g(x))\colon g(x)\in \C \}$ decreases by $\min\{\mathrm{mindeg}_q(g(x))\colon g(x)\in \C \}$.
\end{remark}

\begin{ass}\label{ass:hypf}
Note that there exist $f_1(x),\ldots,f_k(x)\in\C$ such that the following hold:
\begin{enumerate}
    \item $f_1(x),\ldots,f_k(x)$ are monic and $\fqn$-linearly independent;
    \item $M_1:=\deg_q (f_1(x)),\ldots,M_k:=\deg_q (f_k(x))$ are all distinct;
    \item $m_1:={\rm mindeg}_q(f_1(x)),\ldots,m_k:={\rm mindeg}_q(f_k(x))$ are all distinct, and $m_i=0$ for some $i$;
    \item $f_1(x)=x^{q^t}$;
    \item for any $i$, if $f_i(x)$ is a monomial then $m_i=M_i\geq t$.
\end{enumerate}
\end{ass}

Therefore, by Corollary \ref{cor:equiv}, we investigate Moore polynomial sets as in the following definition.
\begin{definition}
A Moore polynomial set $\underline{f}=(f_1(x),\ldots,f_k(x))\subseteq\mathcal{L}_{n,q}^k$ satisfying Assumptions \ref{ass:hypf} is said to be a Moore polynomial set for $q$ and $n$ of index $t$.
\end{definition}

A key tool in our approach is a link between Moore polynomial sets ans  algebraic hypersurface. To this aim, we introduce the following $\fqn$-rational hypersurfaces: $\mathcal{U}:= \mathcal{U}_{\underline{f}}\subset \mathbb{P}^k(\overline{\mathbb{F}}_{q^n})$ is the hypersurface defined by the affine equation
\[
F_{\underline{f}}(X_1,\ldots,X_k):= \det(M_{\underline{f}}(X_1,\ldots,X_k))=0,
\]
and $\mathcal{V}\subset \mathbb{P}^k(\overline{\mathbb{F}}_{q^n})$ is the hypersurface $\mathcal{U}_{(x,x^q,\ldots,x^{q^{k-1}})}$.
Note that
\[
F_{(x,x^q,\ldots,x^{q^{k-1}})}(X_1,\ldots,X_k) = \prod_{(a_1:\ldots:a_k)\in \mathbb{P}^{k-1}(\fq)} (a_1 X_1+\cdots + a_k X_k),
\]
with a suitable choice of the representative for the points $(a_1:\ldots:a_k)$.
Since $f_1(x),\ldots,f_k(x)$ are $\fq$-linearized, the polynomial $F_{(x,x^q,\ldots,x^{q^{k-1}})}(X_1,\dots,X_k)$ divides $F_{\underline{f}}(X_1,\ldots,X_k)$, so that $\mathcal{V}$ is a component of $\mathcal{U}$.
Therefore we can define the $\fqn$-rational variety $\mathcal{W}\subset\mathbb{P}^k(\overline{\mathbb{F}}_{q^n})$ with affine equation
\[
\mathcal{W}\colon \quad F_{\underline{f}}(X_1,\ldots,X_k)/F_{(x,x^q,\ldots,x^{q^{k-1}})}(X_1,\ldots,X_k) =0.
\]
The link between Moore polynomial sets and  algebraic hypersurfaces $\mathcal{W}$ is straightforward.
\begin{theorem}\label{Th:Hyper}
The $k$-tuple $\underline{f}$ is a Moore polynomial set for $q$ and $n$ if and only if all the affine $\mathbb{F}_{q^n}$-rational points of $\mathcal{W}$ lie on $\mathcal{V}$.
\end{theorem}
\begin{proof}
For any $A=\{\alpha_1,\ldots,\alpha_k\}\subseteq\fqn$, the condition $\det(M_{\underline{f},A})=0$ is equivalent to $(\alpha_1,\ldots,\alpha_k)$ being an affine $\fqn$-rational point of $\mathcal{W}$, while the condition $\dim_{\fq}(\langle\alpha_1,\ldots,\alpha_k\rangle_{\fq}) <k$ is equivalent to $(\alpha_1,\ldots,\alpha_k)$ being a point of $\mathcal{V}$. The claim follows.
\end{proof}

In the case when $f_1(x),\ldots,f_k(x)$ are monomials, Theorem \ref{Th:Hyper} was already noticed (using a slightly different terminology) and used in \cite{MR4110235} to prove the following result.

\begin{theorem}\cite[Theorems 1.1, 3.2, 4.1]{MR4110235}
Let $I = \{i_1=0,i_2,\ldots,i_{k}\}$ be a set of non-negative integers with $0 < i_2 < \ldots < i_{k}$ such that $I$ is not in arithmetic progression.
Suppose that one of the following holds:
\begin{itemize}
    \item $|I|=3$ and $n>4 i_{k}+2$;
    \item $|I|>3$, $q>5$ and $n>\frac{13}{3}i_{k}+\log_q(13\cdot 2^{10/3})$.
\end{itemize}
Then $(x,x^{q^{i_2}},\ldots, x^{q^{i_k}})$ is not a Moore polynomial set for $q$ and $n$.
\end{theorem}

In the sequel, we will use the following notation: for any $i=1,\ldots,k$, write $f_i(x)=\sum_{j=m_i}^{M_i}a_{ij}x^{q^j}$ and $f_i(x,z):=\sum_{j=m_i}^{M_i}a_{ij}x^{q^j}z^{q^{M_j}-q^{j}}$.

\subsection{Moore polynomial sets of index 0}

In this section we investigate Moore polynomial sets of index $0$, so that $f_1(x)=x$. Without loss of restriction, we assume $M_1=0<M_2<\cdots< M_k$.

\begin{theorem}\label{Th:progressione}
Suppose that one of the following holds:
\begin{itemize}
    \item $k=3$ and $n>4M_3 +2$;
    \item $k>3$, $q>5$ and $n>\frac{13}{3}M_k + \log_q(13\cdot 2^{10/3})$.
\end{itemize}
If $\underline{f}$ is a Moore polynomial set for $q$ and $n$ of index $0$, then $(M_1=0,M_2,\ldots,M_k)$ is in arithmetic progression and $(m_{\sigma(1)}=0,m_{\sigma(2)},\ldots,m_{\sigma(k)})$ is in arithmetic progression for some $\sigma\in S_k$ with $\sigma(1)=1$.
\end{theorem}
\begin{proof}
In order to prove the claim on the $M_i$'s, consider the intersection $\mathcal{W}_{\infty}=\mathcal{W}\cap\mathcal{H}_{\infty}$ between $\mathcal{W}$ and the hyperplane at infinity $\mathcal{H}_{\infty}\subset \mathbb{P}^{k}(\overline{\mathbb{F}}_{q^n})$.
Note that $\mathcal{W}_{\infty}\subset \mathbb{P}^{k-1}(\overline{\mathbb{F}}_{q^n})$ is defined by
\[
\mathcal{W}_{\infty}\colon\quad
F_{(x,x^{q^{M_2}},\ldots,x^{q^{M_k}})}(X_1,\ldots,X_k)/F_{(x,x^q,\ldots,x^{q^{k-1}})}(X_1,\ldots,X_k) =0.
\]
Suppose that $(M_1,\ldots,M_k)$ is not in arithmetic progression. Then it has been shown in \cite[Theorems 3.1 and 4.2]{MR4110235} that $\mathcal{U}_{\infty}$ contains an $\fqn$-rational non-repeated absolutely irreducible component $\mathcal{X}$.
It follows by Lemma \ref{lemma:bello} that $\mathcal{W}$ has an $\fqn$-rational non-repeated absolutely irreducible component.
Then, as shown in \cite{MR4110235} (in page 9 for $k=3$, and in page 17 for $k>3$), there exists an affine $\fqn$-rational point in $\mathcal{W}\setminus\mathcal{V}$. Thus, $\underline{f}$ is not a Moore polynomial set for $q$ and $n$ by Theorem \ref{Th:Hyper}.

Now suppose that $(m_{\sigma(1)},\ldots,m_{\sigma(k)})$ is not in arithmetic progression for any $\sigma\in S_k$.
Consider the tangent variety $\mathcal{T}$ of $\mathcal{W}$ at the origin $O$. Then
\[
\mathcal{T}\colon\quad
F_{(x^{q^{m_{\sigma(1)}}},x^{q^{m_\sigma(2)}},\ldots,x^{q^{m_\sigma(k)}})}(X_1,\ldots,X_k)/F_{(x,x^q,\ldots,x^{q^{k-1}})}(X_1,\ldots,X_k) =0.
\]
Now the same arguments as above show that $\mathcal{T}$ has an $\fqn$-rational non-repeated absolutely irreducible component.
Then $\mathcal{W}$ has an $\fqn$-rational non-repeated absolutely irreducible component by Lemma \ref{lemma:bello}.
Therefore, as in \cite{MR4110235}, $\mathcal{W}$ has an affine $\fqn$-rational point not in $\mathcal{V}$, so that $\underline{f}$ is not a Moore polynomial set for $q$ and $n$.
\end{proof}

In the rest of this subsection, $\underline{f}$ is a Moore polynomial set for $q$ and $n$ of index $0$, satisfying the assumptions of Theorem \ref{Th:progressione}, so that $M_1=0,M_2=M,\ldots,M_k=(k-1)M$.

Let $\lambda_3,\ldots,\lambda_{k}$ be $\fq$-linearly independent elements of $\fqn$ and define $H_{\underline{f}}(X_1,X_2)=F_{\underline{f}}(X_1,X_2,\lambda_3,\ldots,\lambda_k)\in\fqn[X_1,X_2]$.
Since $\underline{f}$ is a Moore polynomial set, $H_{\underline{f}}(X_1,X_2)\not\equiv0$.
Let $\mathcal{D}_{\underline{f}}\subset\mathbb{P}^{2}(\overline{\mathbb{F}}_{q^n})$ be the curve defined by $H_{\underline{f}}(X_1,X_2)=0$. 
We denote by $H_{\underline{f}}(X_1,X_2,T)$ the homogeneization of  $H_{\underline{f}}(X_1,X_2)$, i.e.
$$H_{\underline{f}}(X_1,X_2,T):=\det\left( 
\begin{array}{cccc}
X_1 & f_2(X_1,T) & \cdots & f_{k}(X_1,T)  \\
X_2 & f_2(X_2,T) & \cdots & f_{k}(X_2,T)  \\
\lambda_3 T  & f_2(\lambda_3) T^{q^M} & \cdots & f_k(\lambda_3)T^{q^{(k-1)M}}  \\
\vdots & \vdots & \ddots & \vdots  \\
\lambda_k T  & f_2(\lambda_k) T^{q^M} & \cdots & f_k(\lambda_k)T^{q^{(k-1)M}}  \\
\end{array}
\right)\;\Bigg/\; T^{\frac{q^{(k-2)M}-1}{q^M-1}}.$$

\begin{lemma}\label{lemma:sali}
If $\mathcal{D}_{\underline{f}}$ has a non-repeated $\fqn$-rational absolutely irreducible component not contained in $\mathcal{D}_{x,x^q,\ldots,x^{q^{k-1}}}$, then $\mathcal{W}$ has a non-repeated $\fqn$-rational absolutely irreducible component not contained in $\mathcal{V}$.
\end{lemma}

\begin{proof}
Consider the variety $\mathcal{W}_3\subset\mathbb{P}^{3}(\overline{\mathbb{F}}_{q^n})$ defined by
\[
\mathcal{W}_3\colon \quad F_{\underline{f}}(X_1,X_2,X_3,\lambda_4,\ldots,\lambda_k)/F_{(x,x^q,\ldots,x^{q^{k-1}})}(X_1,X_2,X_3,\lambda_4,\ldots,\lambda_k) =0.
\]
Let $\Pi_3\subset\mathbb{P}^{3}(\overline{\mathbb{F}}_{q^n})$ be the hyperplane with affine equation $X_3=\lambda_3$.
By the assumptions, $\mathcal{W}_3\cap\Pi_3$ has a non-repeated $\fqn$-rational absolutely irreducible component. Hence, by Lemma \ref{le:subvarieties}, $\mathcal{W}_3$ has a non-repeated $\fqn$-rational absolutely irreducible component.
The claim follows by repeatedly applying this argument to $\mathcal{W}_3,\ldots,\mathcal{W}_k=\mathcal{W}$.
\end{proof}


\begin{remark}
It is readily seen that $(x,f(x))$ is a Moore polynomial set for $q$ and $n$ if and only if $f(x)$ is scattered of index $0$ over $\mathbb{F}_{q^n}$.
By the results in \cite{MR4190573,MR3812212}, if $n>4\deg_{q}(f(x))$ and $(x,f(x))$ is a Moore polynomial set for $q$ and $n$, then $f(x)$ is a monomial with $\gcd(n,\deg_{q}(f(x)))=1$.
Next results deal with the case $k>2$.
\end{remark}

\begin{theorem}\label{Prop:f_monomio}
Let $\underline{f}=(f_1(x)=x,f_2(x),f_3(x))$ be a Moore polynomial set for $q$ and $n$ of index $0$ with $0<M_2<M_3$. 
If $n>4M_3+2$, then $f_2(x)\in\mathcal{L}_{n,q}$ is scattered of index $0$.
\end{theorem}
\begin{proof}
By Theorem \ref{Th:progressione}, $M_3=2M$ where $M=M_2$.
Suppose that $f_2(x)\in\mathcal{L}_{n,q}$ is not scattered of index $0$, so that there exist $\lambda,\mu\in\fqn^*$ such that $\lambda/\mu\notin\fq$ and $f_2(\lambda)/\lambda=f_2(\mu)/\mu$.
By \cite[Corollary 3.4]{bartoli2020r}, we can assume that $\mu\notin\mathbb{F}_q$ and $f_2(\lambda)\ne0$.

Let $\lambda_3=\lambda$ and define $\mathcal{D}_{\underline{f}}$ as above.
Let $\mathcal{D}_{\underline{f}}^\prime$ be the image of $\mathcal{D}_{\underline{f}}$ under the $\fqn$-rational projectivity $\varphi:(X_1\colon X_2\colon T)\mapsto(T\colon X_2-X_1 \colon X_1)$. Note that the point $P=(1\colon 1\colon 0)\in\mathcal{D}_{\underline{f}}$ is mapped by $\varphi$ to $O=(0\colon0\colon1)$.
The curve $\mathcal{D}_{\underline{f}}^\prime$ has affine equation $H^{\prime}_{\underline{f}}(X_1,X_2)=0$, where
\[
H^{\prime}_{\underline{f}}(X_1,X_2)=H_{\underline{f}}(1,X_2+1,X_1)=\det\left( 
\begin{array}{ccc}
1 & f_2(1,X_1)  & f_3(1,X_1)  \\
X_2+1 & f_2(1,X_1)+f_2(X_2,X_1) & f_3(1,X_1) +f_3(X_2,X_1) \\
\lambda & f_2(\lambda)X_1^{q^M-1} &  f_3(\lambda)X_1^{q^{2M}-1}
\end{array}
\right)
\]
\[
=\det\left( 
\begin{array}{ccc}
1 & f_2(1,X_1)  & f_3(1,X_1)  \\
X_2 & f_2(X_2,X_1) & f_3(X_2,X_1) \\
\lambda & f_2(\lambda)X_1^{q^M-1} &  f_3(\lambda)X_1^{q^{2M}-1}
\end{array}
\right)=-\lambda f_2(X_2,X_1)+f_2(\lambda) X_2X_1^{q^M-1} + G(X_1,X_2),
\]
for some  $G(X_1,X_2)\in\fqn[X_1,X_2]$ of degree bigger than $q^M$.

The homogeneous polynomial $L(X_1,X_2)=-\lambda f_2(X_2,X_1)+f_2(\lambda) X_2X_1^{q^M-1}$ has $X_2-\mu X_1$ as a non-repeated factor in $\fqn[X_1,X_2]$, since $\mu$ is a root of the separable polynomial $L(1,X_2)\in\fqn[X_2]$.
Therefore $\mathcal{D}_{\underline{f}}^\prime$ has a non repeated $\fqn$-rational absolutely irreducible component, and the same holds for $\mathcal{D}_{\underline{f}}$.
Since $\mu\notin\mathbb{F}_q$, such a component of $\mathcal{D}_{\underline{f}}$ is not contained in $\mathcal{D}_{(x,x^{q},x^{q^2})}$.

By Lemma \ref{lemma:sali}, $\mathcal{W}$ has a non-repeated $\fqn$-rational absolutely irreducible component $\mathcal{Z}$ not contained in $\mathcal{V}$.
The degree of $\mathcal{V}$ is $q^2+q+1$, and the degree of $\mathcal{Z}$ is at most $q^{2M}+q^{M}-q^2-q$. Thus, by \cite[Corollary 7]{SLAVOV201760}, $\mathcal{Z}$ has an affine $\fqn$-rational point not on $\mathcal{V}$, a contradiction to Theorem \ref{Th:Hyper}.
\end{proof}

\begin{theorem}\label{Prop:h=2}
Let $\underline{f}=(f_1(x)=x,f_2(x),f_3(x))$ be a Moore polynomial set for $q$ and $n$ of index $0$ such that $0<M_2<M_3$, and $(q,M)\notin\{(2, 2),(2, 4),( 3, 2),( 4, 2),( 5, 2)\}$.
If $n>4M_3+2$, then $\underline{f}=(x,x^{q^M},x^{q^{2M}})$ with $\gcd(M,n)=1$.
\end{theorem}

\begin{proof}
By Theorem \ref{Th:progressione}, $M_3=2M$ and $\max\{m_2,m_3\}=2\min\{m_2,m_3\}$.
By Theorem \ref{Prop:f_monomio}, $f_2(x)$ is scattered of index $0$ over $\fqn$. Thus, by the numerical assumption on $n$, it follows that $f_2(x)=q^M$ and $\gcd(M,n)=1$; see \cite[Section 3.1]{MR3812212} for $q>5$ and \cite[Section 5]{MR4190573} for $q\leq5$.

From $m_2=M$ it follows that $m_3\in\{2M,M/2\}$. Suppose by contradiction that $f_3(x)\ne x^{q^{2M}}$, so that $m_3=M/2<M$, and in particular $M$ is even.
Choose $\lambda_3=\lambda\in\mathbb{F}_{q^n}^*$ such that $f_2(\lambda)f_3(\lambda)\ne0$.
Via Theorem \ref{cor:lowergrado}, we will prove that the variety $\mathcal{W}_3$ with affine equation
\[
\mathcal{W}_3 \colon\quad F_{\underline{f}}(X_1,X_2,\lambda)/F_{(x,x^{q},x^{q^2})}(X_1,X_2,\lambda)=0
\]
has a non-repeated $\fqn$-rational absolutely irreducible component not contained in $\mathcal{V}$.

Suppose that $\mathcal{W}_{3}$ splits into two components $\mathcal{A}$ and $\mathcal{B}$ sharing no common absolutely irreducible component. Let $\tilde{\mathcal{A}}$ and $\tilde{\mathcal{B}}$ be two components of $\mathcal{D}_{\underline{f}}$ sharing no common absolutely irreducible components and such that $\mathcal{A}\subseteq\tilde{\mathcal{A}}$, $\mathcal{B}\subseteq\tilde{\mathcal{B}}$.
Singular points of $\mathcal{W}_3$ are also singular points of $\mathcal{D}_{\underline{f}}$, and the intersection multiplicity of $\mathcal{A}$ and $\mathcal{B}$ at a point is at most the intersection multiplicity of $\tilde{\mathcal{A}}$ and $\tilde{\mathcal{B}}$ at that point.
We start the inspection of singular points of $\mathcal{W}_{3}$ from affine ones.
Let $P=(\alpha,\beta)\in\mathbb{P}^{2}(\overline{\mathbb{F}}_{q^n})$ be an affine point of $\mathcal{D}_{\underline{f}}$.
The point $P$ is singular for $\mathcal{D}_{\underline{f}}$ if and only if $f_3(\lambda)f_2(\alpha)-f_2(\lambda)f_3(\alpha)=0$ and $f_3(\lambda)f_2(\beta)-f_2(\lambda)f_3(\beta)=0$, that is
\begin{equation*}
f_3(\lambda)\alpha^{q^M}-\lambda^{q^M}f_3(\alpha)=f_3(\lambda)\beta^{q^M}-\lambda^{q^M}f_3(\beta)=0.
\end{equation*}
Also, the intersection multiplicity of $\tilde{\mathcal{A}}$ and $\tilde{\mathcal{B}}$ at $P$ equals the intersection multiplicity of $\tau(\tilde{\mathcal{A}})$ and $\tau(\tilde{\mathcal{B}})$ at $\tau(P)=O=(0,0)$, where $\tau$ is the translation $(X_1,X_2)\mapsto(X_1-\alpha,X_2-\alpha)$. The image  $\mathcal{D}_{\underline{f}}^\prime$ of $\mathcal{D}_{\underline{f}}$ under $\tau$ has affine equation $H_{\underline{f}}^\prime(X_1,X_2)=0$, with
\[
H^{\prime}_{\underline{f}}(X_1,X_2) \,=\, H_{\underline{f}}(X_1+\alpha,X_2+\beta)\,=\, \det\left( 
\begin{array}{ccc}
X_1+\alpha & X_1^{q^M}+\alpha^{q^M}  & f_3(X_1)+f_3(\alpha)  \\
X_2+\beta & X_2^{q^M}+\beta^{q^M} & f_3(X_2)+g(\beta)  \\
\lambda & \lambda^{q^M} &  f_3(\lambda)
\end{array}
\right)
\]
\[
=\,a\Big((\lambda \alpha^{q^M}-\alpha   \lambda^{q^M})X_2^{q^{m_3}}-(\lambda  \beta^{q^M}-\beta   \lambda^{q^M})X_1^{q^{m_3}}-\lambda^{q^M}(X_1X_2^{q^{m_3}}-X_2X_1^{q^{m_3}})\Big)+G(X_1,X_2),
\]
where $a\ne0$ is the coefficient of $x^{q^{m_3}}$ in $f_3(x)$, and $G(X_1,X_2)\in\fqn[X_1,X_2]$ has degree bigger than $q^{m_3}+1$.
We denote respectively by $(H^{\prime}_{\underline{f}})_{q^{m_3}}$ and $(H^{\prime}_{\underline{f}})_{q^{m_3}+1}$ the homogeneous polynomials $(\lambda \alpha^{q^M}-\alpha   \lambda^{q^M})X_2^{q^{m_3}}-(\lambda  \beta^{q^M}-\beta   \lambda^{q^M})X_1^{q^{m_3}}$ and $X_1X_2^{q^{m_3}}-X_2X_1^{q^{m_3}}$. If non-vanishing, they are, up to a scalar multiple, the homogeneous parts of smallest degrees in $H_{\underline{f}}^{\prime}(X_1,X_2)$.
Note that the $q^{m_3}+1$ linear factors of $(H^{\prime}_{\underline{f}})_{q^{m_3}+1}$ are all distinct.
\begin{itemize}
    \item There are at most $q^{2M}$ singular points $(\alpha,\beta)$ of $\mathcal{D}_{\underline{f}}$ which satisfy $\lambda \alpha^{q^M}-\alpha  \lambda^{q^M}= \lambda \beta^{q^M}-\beta  \lambda^{q^M}=0$. In this case, $(H^{\prime}_{\underline{f}})_{q^{m_3}+1}$ is the non-zero homogeneous part of smallest degree in $H_{\underline{f}}^\prime(X_1,X_2)$. Thus $O$ is an ordinary $(q^{m_3}+1)$-fold point for $\mathcal{D}_{\underline{f}}^\prime$, and by Lemma \ref{le:intersection_number_ordinary} the intersection multiplicity of $\tau(\tilde{\mathcal{A}})$ and $\tau(\tilde{\mathcal{B}})$ at $O$ is at most $(q^{m_3}+1)^2/4$.
    
    \item There are at most $2(q^{2M-m_3}-1)\cdot q^M$ singular points $(\alpha,\beta)$ of $\mathcal{D}_{\underline{f}}$ which satisfy either $\lambda \alpha^{q^M}-\alpha  \lambda^{q^M}\neq  0=\lambda \beta^{q^M}-\beta  \lambda^{q^M}$ or $\lambda \alpha^{q^M}-\alpha  \lambda^{q^M}=  0\neq \lambda \beta^{q^M}-\beta  \lambda^{q^M}$.

    In this case, $(H^{\prime}_{\underline{f}})_{q^{m_3}}=X_2^{q^{m_3}}$ or $(H^{\prime}_{\underline{f}})_{q^{m_3}}=X_1^{q^{m_3}}$ up to a non-zero scalar multiple, and hence $\gcd(H^{\prime}_{\underline{f}})_{q^{m_3}},(H^{\prime}_{\underline{f}})_{q^{m_3}+1}=X_2$ or $X_1$.
    By Lemma \ref{le:intersection_number_linear_term}, the intersection multiplicity of $\tau(\tilde{\mathcal{A}})$ and $\tau(\tilde{\mathcal{B}})$ at $O$ is at most $q^{m_3}$.
    
    \item There are at most $(q^{2M-m_3}-1)\cdot (q^{m_3}-1)\cdot q^M$ singular points $(\alpha,\beta)$ of $\mathcal{D}_{\underline{f}}$ which satisfy $\lambda \alpha^{q^M}-\alpha  \lambda^{q^M}\neq  0$, $\lambda \beta^{q^M}-\beta  \lambda^{q^M}\neq 0$, and $\eta (\lambda \alpha^{q^M}-\alpha   \lambda^{q^M})=\xi (\lambda  \beta^{q^M}-\beta   \lambda^{q^M})$ for some $(\xi:\eta) \in \mathbb{P}^{1}(\mathbb{F}_{q^{m_3}})\setminus\{(1:0),(0:1)\}$.
    In this case, $(H^{\prime}_{\underline{f}})_{q^{m_3}}=(\xi X_2 - \eta X_1)^{q^{m_3}}$ up to a non-zero scalar multiple, and hence $(H^{\prime}_{\underline{f}})_{q^{m_3}}$ and $(H^{\prime}_{\underline{f}})_{q^{m_3}+1}$ are not coprime.
    By Lemma \ref{le:intersection_number_linear_term}, the intersection multiplicity of $\tau(\tilde{\mathcal{A}})$ and $\tau(\tilde{\mathcal{B}})$ at $O$ is at most $q^{m_3}$.
    
    \item If a singular point $(\alpha,\beta)$ of $\mathcal{D}_{\underline{f}}$ satisfies $\eta (\lambda \alpha^{q^M}-\alpha   \lambda^{q^M})=\xi (\lambda  \beta^{q^M}-\beta   \lambda^{q^M})$ for some $(\xi:\eta) \notin \mathbb{P}^{1}(\mathbb{F}_{q^{m_3}})$, then $(H^{\prime}_{\underline{f}})_{q^{m_3}}$ and $(H^{\prime}_{\underline{f}})_{q^{m_3}+1}$ are coprime. In this case, by Lemma \ref{le:intersection_number_m_m1_coprime}, the intersection multiplicity of $\tau(\tilde{\mathcal{A}})$ and $\tau(\tilde{\mathcal{B}})$ at $O$ is $0$.
\end{itemize}
Since the homogeneus part of largest degree in $H_{\underline{f}}(X_1,X_2)$ is
\[\lambda\cdot  X_1^{q^M}\cdot \prod_{\gamma\in\mathbb{F}_{q^M}}(X_2-\gamma X_1)^{q^M},\]
the points at infinity of $\mathcal{D}_{\underline{f}}$ are $(0\colon1\colon0)$ and $(1\colon\gamma\colon0)$ with $\gamma \in\mathbb{F}_{q^M}$.
As the map $(X_1\colon X_2\colon T)\mapsto(X_2\colon X_1\colon T)$ maps $(0\colon1\colon0)$ to $(1\colon0\colon0)$ and leaves invariant the curves $\mathcal{D}_{\underline{f}}$ and $\mathcal{D}_{(x,x^q,x^{q^2})}$, it is enough to consider the points $P_{\gamma}=(1\colon\gamma\colon0)$ with $\gamma\in\mathbb{F}_{q^M}$.
The intersection multiplicity of $\tilde{\mathcal{A}}$ and $\tilde{\mathcal{B}}$ at $P_{\gamma}$ equals the intersection multiplicity of $\sigma(\tilde{\mathcal{A}})$ and $\sigma_{\gamma}(\tilde{\mathcal{B}})$ at $\sigma_{\gamma}(P)=O=(0,0)$, where $\sigma_{\gamma}\colon (X_1\colon X_2\colon T)\mapsto(T\colon X_2-\gamma T\colon X_1)$. The image  $\mathcal{D}_{\underline{f}}^{\prime\prime}$ of $\mathcal{D}_{\underline{f}}$ under $\sigma_{\gamma}$ has affine equation $H_{\underline{f}}^{\prime\prime}(X_1,X_2)=0$, where
{\small
\[
H^{\prime\prime}_{\underline{f}}(X_1,X_2)\,=\,H_{\underline{f}}(1,X_2+\alpha,X_1)\,=\, \det\left( 
\begin{array}{ccc}
1 & f_2(1,X_1)  & f_3(1,X_1)   \\
X_2+\gamma & f_2(X_2+\gamma,X_1) & f_3(X_2+\gamma,X_1) \\
\lambda & f_2(\lambda)X_1^{q^M-1} &  f_3(\lambda)X_1^{q^{2M}-1}
\end{array}
\right) \,=
\]
\[
\det\left( 
\begin{array}{ccc}
1 & 1  & f_3(1,X_1)  \\
X_2 & X_2^{q^M} & f_3(X_2,X_1)+f_3(\gamma,X_1)-\gamma f_3(1,X_1) \\
\lambda & \lambda^{q^M}X_1^{q^M-1} &  f_3(\lambda)X_1^{q^{2M}-1}
\end{array}
\right)
= \lambda^{q^M} X_2X_1^{q^M-1} -\lambda X_2^{q^M}+G(X_1,X_2),
\]}
for some  $G(X_1,X_2)$ of degree greater than $q^M$ (here, we used that the constant term in $f_3(\gamma,X_1)$ is the same as in $\gamma f_3(1,X_1)$).
Since $\lambda^{q^M} X_2X_1^{q^M-1} -\lambda X_2^{q^M}$ is homogeneous and separable in each variable, $0$ is an ordinary $q^M$-fold point for $\mathcal{D}_{\underline{f}}^{\prime\prime}$, and by Lemma \ref{le:intersection_number_ordinary} the intersection multiplicity of $\tilde{\mathcal{A}}$ and $\tilde{\mathcal{B}}$ at $P_{\alpha}$ is at most $q^{2M}/4$. The same holds at $(0\colon1\colon0)$.

Summing up, the number of intersection points of two components of $H_{\underline{f}}(X_1,X_2,T)=0$, counted with multiplicity, satisfies
\[
\sum_{P}I(P,\mathcal{A}\cap\mathcal{B})\leq
q^{2M}\frac{(q^{m_3}+1)^2}{4}+2(q^{2M-m_3}-1) q^{M+m_3} +(q^{2M-m_3}-1) (q^{m_3}-1) q^{M+m_3}
+(q^M+1)\frac{q^{2M}}{4}.
\]
Since $(q,M)\notin\{(2, 2),(2, 4),( 3, 2),( 4, 2),( 5, 2)\}$, the above quantity is less than  
$$\frac{2}{9}\deg(\mathcal{W}_3)^2=\frac{2}{9}(q^{2M}+q^M-q^2-q)^2.$$
By Theorem \ref{cor:lowergrado}, $\mathcal{W}_3$ contains an $\fqn$-rational absolutely irreducible component $\mathcal{X}$.
Note that $\mathcal{D}_{\underline{f}}$ has only finitely many singular points,
and hence $\mathcal{X}$ is non-repeated and not contained in $\mathcal{V}$.
Arguing as in the last paragraph of the proof of Theorem \ref{Prop:f_monomio}, a contradiction arises.
This shows $m_3=2M$, i.e. $f_3(x)=x^{q^{2M}}$.

\end{proof}




By means of an induction argument, we are able to extend the result of Theorem \ref{Prop:h=2} to any Moore polynomial set of index $0$, as follows.

\begin{theorem}\label{prop:f2_scattered}
Let $\underline{f}=(f_1(x)=x,f_2(x),\ldots, f_k(x))$, with $k>3$, be a Moore polynomial set for $q$ and $n$ of index $0$ such that  
$0<M_2<\cdots<M_k$.
Suppose also that $q>5$ and $n>\frac{13}{3}M_k + \log_{q}(13\cdot2^{10/3})$.
Then $\underline{f}=(x,x^{q^M},\ldots, x^{q^{(k-1)M}})$ with $\gcd(M,n)=1$.
\end{theorem}

\begin{proof}
By Theorem \ref{Th:progressione}, $M_i=(i-1)M$ for every $i$, with $\gcd(M,n)=1$. Also, $\{0,m_2,\ldots,m_{k}\}$ can be ordered so that they are in arithmetic progression.

We prove by finite induction on $i\in\{3,\ldots,k\}$ the following fact: if $\underline{h}=(x,f_2(x),\ldots,f_i(x))$ satisfies $\underline{h}\ne (x,x^{q^{M}},\ldots,x^{q^{(i-1)M}})$, then the hypersurface $\mathcal{U}_{\underline{h}}$ has a non-repeated $\fqn$-rational absolutely irreducible component not contained in $\mathcal{V}$.
The base $i=3$ has been worked out in the proof of Theorem \ref{Prop:h=2}.
For $i>3$, define the map $\varphi\colon(X_1\colon \ldots\colon X_i\colon T)\mapsto(T\colon X_2-X_1\colon X_3\colon\ldots X_i\colon X_1)$, which maps $(1\colon1\colon0\colon\ldots\colon0)\in\mathcal{U}_{\underline{f}}$ to  $O=(0\colon\ldots\colon0\ldots1)$, and consider the image $\mathcal{U}_{\underline{f}}^\prime$ of $\mathcal{U}_{\underline{f}}$ under $\varphi$, which has affine equation $F_{\underline{f}}^\prime(X_1,\ldots,X_i)=0$, where $F^{\prime}_{\underline{f}}(X_1,X_2,\ldots,X_i)$ equals
{\scriptsize \[
\det\left( 
\begin{array}{cccc}
1 & f_2(1,X_1)  &   \ldots & f_i(1,X_1)  \\
X_2+1 & f_2(1,X_1)+f_2(X_2,X_1)  &   \ldots & f_i(1,X_1) +f_i(X_2,X_1)  \\
X_3&f_2(X_3,X_1)  &   \ldots & f_i(X_3,X_1)  \\
\vdots &\vdots& \ddots &\vdots\\
X_i&f_2(X_i,X_1)  &   \ldots & f_k(X_i,X_1)  \\
\end{array}
\right)
=\det\left( 
\begin{array}{cccc}
1 & f_2(1,X_1)  &   \ldots & f_i(1,X_1)  \\
X_2& f_2(X_2,X_1)  &   \ldots & f_i(X_2,X_1)  \\
X_3&f_2(X_3,X_1)  &   \ldots & f_i(X_3,X_1)  \\
\vdots &\vdots& \ddots &\vdots\\
X_i&f_2(X_i,X_1)  &   \ldots & f_i(X_i,X_1)  \\
\end{array}
\right).\]}
The tangent cone to $\mathcal{U}_{\underline{f}}^\prime$ at $O$ has equation
\[
F_{\underline{g}}^*(X_2,\ldots,X_i,X_1)=
\det\left( 
\begin{array}{cccc}
X_2& f_2(X_2,X_1)  &   \ldots & f_{i-1}(X_2,X_1)  \\
X_3&f_2(X_3,X_1)  &   \ldots & f_{i-1}(X_3,X_1)  \\
\vdots &\vdots& \ddots &\vdots\\
X_i&f_2(X_i,X_1)  &   \ldots & f_{i-1}(X_i,X_1)  \\
\end{array}
\right)=0,\]
where $\underline{g}=(f_1(x)=x,f_2(x),\ldots, f_{i-1}(x))$. 
Note that $F_{\underline{g}}^*(X_2,\ldots,X_i,X_1)$ is homogeneous, and its dehomogenized polynomial with respect to $X_1$ is $F_{\underline{g}}(X_2,\ldots,X_i)$.

If $\underline{g}\ne(x,x^{q^M},\ldots,x^{q^{(i-2)M}})$, then by induction hypothesis $\mathcal{U}_{\underline{g}}$ has a non-repeated $\fqn$-rational absolutely irreducible component not contained in $\mathcal{V}$, and hence by Lemma \ref{lemma:bello} the same holds for $\mathcal{U}_{\underline{h}}$.
If $\underline{g}=(x,x^{q^M},\ldots,x^{q^{(i-2)M}})$ then $f_i(x)=x^{q^{(i-1)M}}$, because $i\geq4$ implies that the arithmetic progressions of the $M_j$'s and $m_j$'s both have ratio $M$.

For $i=k$, if $\underline{f}\ne(x,x^{q^M},\ldots,x^{q^{(k-2)M}})$, then $\mathcal{W}$ has a non-repeated $\fqn$-rational absolutely irreducible component $\mathcal{Z}$ not contained in $\mathcal{V}$. Thus, by \cite[Corollary 7]{SLAVOV201760}, $\mathcal{Z}$ has an affine $\fqn$-rational point not on $\mathcal{V}$, a contradiction to Theorem \ref{Th:Hyper}.
\end{proof}

\subsection{Moore polynomial sets of positive index}

In this section we investigate Moore polynomial sets of index $t>0$, so that $f_1(x)=x^{q^t}$.
\begin{proposition}\label{prop:posindexord}
Suppose that one of the following holds:
\begin{itemize}
    \item $k=3$ and $n>4M_3+2$;
    \item $k>3$, $q>5$ and $n>\frac{13}{3}M_k + \log_{q}(13\cdot 2^{10/3})$.
\end{itemize}
If $\underline{f}$ is a Moore polynomial set for $q$ and $n$ of index $t$, then $(m_{\sigma(1)},\ldots,m_{\sigma(k)})$ is in arithmetic progression for some $\sigma\in S_n$.
\end{proposition}
\begin{proof}
Since $m_i=0$ for some $i$, the proof is the same as in the proof of Theorem \ref{Th:progressione} for $t=0$.
\end{proof}
Up to reordering, we can assume that the permutation $\sigma$ in Proposition \ref{prop:posindexord} satisfies $\sigma(1)=2$, that is, $f_2(x)$ is separable.


\begin{proposition}\label{Prop:f_monomio_index}
Let $\underline{f}=(f_1(x)=x^{q^t},f_2(x),f_3(x))$ be a Moore polynomial set for $q$ and $n$ of index $t>0$ such that $f_2(x)$ is separable.
If $\max\{t,M_2\}<M_3$ and $n>4M_3+2$, then $f_2(x)\in\mathcal{L}_{n,q}$ is scattered of index $t$.
\end{proposition}
\begin{proof}
The proof is similar to the one of Theorem \ref{Prop:f_monomio}. 
Suppose that $f_2(x)\in\mathcal{L}_{n,q}$ is not scattered of index $t$.
Then there exist $\lambda,\mu\in\mathbb{F}_{q^n}^*$ such that $\mu\notin\mathbb{F}_q$, $\lambda/\mu\notin\mathbb{F}_{q}$ and $f_2(\lambda)/\lambda^{q^t}=f_2(\mu)/\mu^{q^t}\ne0$.
Let $\lambda_3=\lambda$ and define $\mathcal{D}_{\underline{f}}$ as above.
Then $\mathcal{D}_{\underline{f}}$ is ${\rm PGL}(3,q^n)$-equivalent to the curve $\mathcal{D}_{\underline{f}}^\prime$ with affine equation $H_{\underline{f}}^\prime(X_1,X_2)=0$, where
\[
H^{\prime}_{\underline{f}}(X_1,X_2) = H_{\underline{f}}(1,X_2+1,X_1) = -\lambda^{q^t} f_2(X_2,X_1)X_1^{{q^t}-1}+f_2(\lambda) X_2^{q^t}X_1^{q^{M_2}-1} + G(X_1,X_2)
\]
and $G(X_1,X_2)$ has degree at least $q^t+q^{M_2}$.
The tangent cone to $\mathcal{D}_{\underline{f}}^{\prime}$ at $(0,0)$ has a non-repeated $\fqn$-rational absolutely irreducible component with affine equation $X_2-\mu X_1=0$, which corresponds to a non-repeated $\mathbb{F}_{q^n}$-rational absolutely irreducible component of $\mathcal{D}_{\underline{f}}$ which is not contained in $\mathcal{D}_{(x,x^q,x^{q^2})}$.
Arguing as in the proof of Theorem \ref{Prop:f_monomio}, the claim follows.
\end{proof}

We now use the known classification results on exceptional scattered polynomials.

\begin{comment}
\begin{verbatim}

for q in [2,3,4,5,7,8,9,11] do
    for M in [2,4,6,8,10] do
          m3 := M div 2;
          UB := Ceiling(2*(q^(2*M)+q^M-q^2-q)^2/9 );
          LB := q^(2*M)*Floor((q^m3+1)^2/4)
               +2*(q^(2*M-m3)-1)*q^(M+m3)
               +(q^(2*M-m3)-1)*(q^m3-1)*q^(M+m3)
               +(q^M+1)*Floor(q^(2*M)/4 );
        UB := (2*(q^(2*M)+q^M-q^2-q)^2/9 );
          LB := q^(2*M)*((q^m3+1)^2/4)
               +2*(q^(2*M-m3)-1)*q^(M+m3)
               +(q^(2*M-m3)-1)*(q^m3-1)*q^(M+m3)
               +(q^M+1)*(q^(2*M)/4 );
          if LB ge UB then 
                [q,M];
          end if;
    end for;
end for;
\end{verbatim}
\end{comment}

\begin{corollary}
Let $\underline{f}=(f_1(x)=x^{q^t},f_2(x),f_3(x))$ be a Moore polynomial set for $q$ and $n$ of index $t$ such that $f_2(x)$ is separable, $\max\{t,M_2\}<M_3$, and $n>4M_3+2$. Then:
\begin{enumerate}
    \item $t>0$ and $\max\{t,M_2\}$ is not an odd prime;
    \item if either $t=1$, or $t=2$ and $q$ is odd, then   $\underline{f}=(x^{q^t},ax+x^{q^{2t}},f_3(x))$ and $m_3=2t$.
\end{enumerate}
\end{corollary}
\begin{proof}
Since $f_2(x)$ is separable, the case $t=0$ cannot occur by definition. Then $f_2(x)\in\mathcal{L}_{n,q}$ is exceptional scattered of positive index $t$.
\begin{enumerate}
    \item If $\max\{t,M_2\}$ is an odd prime, then from \cite[Theorem 1.4]{MR4163074} it follows $f_2(x)=x$, so that $\underline{f}$ has index $0$, a contradiction.
\item If either $t=1$, or $t=2$ and $q$ is odd, then the results in \cite[Page 511]{MR3812212} and \cite[Theorem 1.4 and Corollary 1.5]{MR4190573} imply $f_2(x)=ax+x^{q^{2t}}$ with $a\ne0$. The claim follows from Proposition \ref{prop:posindexord}.
\end{enumerate}
\end{proof}




Proposition \ref{Prop:f_monomio_index} can be extended as follows.

\begin{theorem}\label{Th:Main2}
Let $\underline{f}=(f_1(x)=x^{q^t},f_2(x),\ldots,f_k(x))$, with $k>3$, be a Moore polynomial set for $q$ and $n$ of index $t>0$ such that $f_2(x)$ is separable.
Suppose also that $q>5$, $\max\{t,M_2\}<M_i$ for any $i\geq3$, and $n>\frac{13}{3}\max\{M_i\colon i\geq 3\}+\log_{q}(13\cdot2^{10/3})$. Then $f_2(x)\in\mathcal{L}_{n,q}$ is scattered of index $t$.
\end{theorem}
\begin{proof}
It can be proved by finite induction on $i\in\{3,\ldots,k\}$ that, if $\underline{h}=(x^{q^t},f_2(x),\ldots,f_i(x))$ and $f_2(x)\in\mathcal{L}_{n,q}$ is not scattered of index $t$, then the hypersurface $\mathcal{U}_{\underline{h}}$ has a non-repeated $\fqn$-rational absolutely irreducible component not contained in $\mathcal{V}$. The base $i=3$ is in the proof of Proposition \ref{Prop:f_monomio_index}. For $i>3$, the argument is analogous to the one in the proof of Theorem \ref{prop:f2_scattered}.
The claim then follows again by using \cite[Corollary 7]{SLAVOV201760}.
\end{proof}

Recalling the correspondence between Moore polynomial sets and MRD codes described in Corollary \ref{cor:equiv}, we finally obtain Main Theorem as a consequence of Theorem \ref{Prop:h=2}, Theorem \ref{prop:f2_scattered}, Proposition \ref{Prop:f_monomio_index} and Theorem \ref{Th:Main2}.

Note that, if the hypothesis of $n$ being large enough in the aforementioned results are incorporated in the assumptions of Main Theorem, then the exceptionality of the MRD code $\C\subset\mathcal{L}_{n,q}$ can be dropped, as well as the exceptionality of the scattered property for $f_2(x)\in\mathcal{L}_{n,q}$.

\section{Known examples of Moore polynomial sets}

This section is devoted to the description of the known examples of Moore polynomial sets corresponding to inequivalent $\fqn$-linear MRD codes; see Table \ref{kMPS}. The only known examples of exceptional Moore polynomial sets are the first two in Table \ref{kMPS}.

Let $b \colon \mathcal{L}_{n,q}\times \mathcal{L}_{n,q}\rightarrow \mathbb{F}_q$ be the bilinear form given by $ b(f,g)=\mathrm{Tr}_{q^n/q}\left( \sum_{i=0}^{n-1} a_ib_i \right),$
where $\mathrm{Tr}_{q^n/q}(x)=\sum_{i=0}^{n-1}x^{q^i}$, $f(x)=\sum_{i=0}^{n-1}a_ix^{q^i}, g(x)=\sum_{i=0}^{n-1}b_ix^{q^i} \in \mathcal{L}_{n,q}$.
The \emph{Delsarte dual code} of a rank metric code $\mathcal{C}\subseteq \mathcal{L}_{n,q}$ is
\[ \mathcal{C}^\perp=\{ f(x) \in \mathcal{L}_{n,q} \colon b(f,g)=0,\,\, \text{for all}\,\, g(x) \in \mathcal{C} \}.  \]
Recall that the  Delsarte dual code of an MRD code, having minimum distance greater than one, is an MRD code; see e.g. \cite{delsarte1978bilinear,gabidulin1985theory}. This yields new examples of Moore polynomial sets; see lines 4,6,8,10,12,14 in Table \ref{kMPS}.

\tabcolsep=0.2 mm
{\small
\begin{longtable}{|c|c|c|c|c|}
\caption{Known examples of Moore polynomial sets}\label{kMPS}\\
 \hline
$n$ & $k$ & $f_1(x),\ldots,f_k(x)$ & \mbox{conditions} & \mbox{references} \\
\endfirsthead
\hline
$n$ & $k$ & $f_1(x),\ldots,f_k(x)$ & \mbox{conditions} & \mbox{references} \\
\endhead
\hline
& & $x,x^{q^s},\ldots,x^{q^{s(k-1)}}$ & $\gcd(s,n)=1$ & \cite{delsarte1978bilinear,gabidulin1985theory,kshevetskiy2005new} \\ \hline
 & & $x^{q^s},\ldots,x^{q^{s(k-1)}},x+\delta x^{q^{sk}}$ & \!\!$\begin{array}{cc} \gcd(s,n)=1,\\ \mathrm{N}_{q^n/q}(\delta)\neq (-1)^{nk}\end{array}$ \!\!& \cite{MR3543528,lunardon2018generalized}\\ \hline
$2t$ & $2$ & \begin{tabular}{c}$x$,\\$x^{q^s}+x^{q^{s(t-1)}}+\delta^{q^t+1}x^{q^{s(t+1)}}+\delta^{1-q^{2t-1}}x^{q^{s(2t-1)}}$\end{tabular} &\!\! $\begin{array}{cc} q \hspace{0.1cm} \text{odd}, \\ \mathrm{N}_{q^{2t}/q^t}(\delta)=-1,\\ \gcd(s,n)=1 \end{array}$ & \cite{MR4173668,longobardi2021linear,longobardi2021large,NPZ,zanella2020vertex}\\ \hline
$2t$ & $2t-2$ & $\begin{array}{cccc} x^{q^{si}} \colon i \notin\{0,1,t-1,t+1,2t-1\},\\ h_1(x)=x^{q^s}-x^{q^{s(t-1)}},\\ h_2(x)=\delta^{q^t+1}x^{q^s}-x^{q^{s(t+1)},}\\ h_3(x)=\delta^{1-q^{2t-1}}x^{q^s}-x^{q^{s(2t-1)}} \end{array}$ & $\begin{array}{cc} q \hspace{0.1cm} \text{odd}, \\ \mathrm{N}_{q^{2t}(q^t}(\delta)=-1,\\ \gcd(s,n)=1 \end{array}$ & \cite{MR4173668,longobardi2021linear,longobardi2021large,NPZ,zanella2020vertex}\\ \hline
$6$ & $2$ & $x,x^q+\delta x^{q^{4}}$  &  $\begin{array}{cc} q>4, \\ \text{certain choices of} \, \delta \end{array}$ & \cite{csajbok2018new,bartoli2020conjecture,polverino2020number} \\ \hline
$6$ & $4$ & $x^q, x^{q^2}, x^{q^4},x-\delta^{q^5} x^{q^{3}}$  &  $\begin{array}{cc} q>4, \\ \text{certain choices of} \, \delta \end{array}$ & \cite{csajbok2018new,bartoli2020conjecture,polverino2020number} \\ \hline
$6$ & $2$ & $x,x^{q}+x^{q^3}+\delta x^{q^5}$ & $\begin{array}{cccc}q \hspace{0.1cm} \text{odd}, \, \delta^2+\delta =1 \end{array}$
 & \cite{csajbok2018new2,marino2020mrd} \\ \hline
 $6$ & $4$ & $x^q,x^{q^3},x-x^{q^2},x^{q^4}-\delta x$ & $\begin{array}{cccc}q \hspace{0.1cm} \text{odd}, \, \delta^2+\delta =1 \end{array}$
 & \cite{csajbok2018new2,marino2020mrd} \\ \hline
$7$ & $3$ & $x,x^{q^s},x^{q^{3s}}$ & $\begin{array}{cc} q \hspace{0.1cm} \text{odd}, \, \gcd(s,7)=1\end{array}$ & \cite{csajbok2020mrd} \\ \hline
$7$ & $4$ & $x,x^{q^{2s}},x^{q^{3s}},x^{q^{4s}}$ & $\begin{array}{cc} q \hspace{0.1cm} \text{odd}, \, \gcd(s,n)=1\end{array}$ & \cite{csajbok2020mrd} \\ \hline
$8$ & $3$ & $x,x^{q^s},x^{q^{3s}}$ & $\begin{array}{cc} q \equiv 1 \pmod{3},\\ \gcd(s,8)=1\end{array}$  & \cite{csajbok2020mrd} \\ \hline
$8$ & $5$ & $x,x^{q^{2s}},x^{q^{3s}},x^{q^{4s}},x^{q^{5s}}$ & $\begin{array}{cc} q \equiv 1 \pmod{3},\\ \gcd(s,8)=1 \end{array}$ & \cite{csajbok2020mrd} \\ \hline
$8$ & $2$ & $x,x^{q}+\delta x^{q^5}$ & $\begin{array}{cc} q\,\text{odd}, \, \delta^2=-1\end{array}$ & \cite{csajbok2020mrd} \\ \hline
$8$ & $6$ & $x^q,x^{q^2},x^{q^3},x^{q^5},x^{q^6},x-\delta x^{q^4}$ & $\begin{array}{cc} q\,\text{odd}, \, \delta^2=-1\end{array}$ & \cite{csajbok2020mrd} \\ \hline
\end{longtable}}

\section{Conclusions and open problems}

In this paper we introduce the notion of \emph{exceptional} linear maximum rank distance codes of a given index, which naturally extends the notion of exceptionality for a scattered polynomial in the rank metric framework.
We then classify those of index $0$, and prove that those of positive index contain an exceptional scattered polynomial of the same index.

We list a couple of open problems related to the obtained results.

\begin{itemize}
    \item Under the assumptions of Proposition \ref{Prop:f_monomio_index} or Theorem \ref{Th:Main2}, for $n$ large enough,  one may conjecture that  
  Moore polynomial sets of positive index do not exist. 
    Whereas, relaxing the assumption $\max\{t,M_2\}<M_i$ for every $i\geq 3$, 
    one should include also the second example listed in Table \ref{kMPS}, that is the one corresponding to generalized twisted Gabidulin codes. 
    However, a new approach seems to be needed. 
   \item A complete classification of exceptional scattered polynomials could yield to more precise results on the asymptotics of Moore polynomial sets of positive index and hence of  $\fqn$-linear MRD codes in $\mathcal{L}_{n,q}$.
\end{itemize}

\section*{Acknowledgments}
This research  was supported by the Italian National Group for Algebraic and Geometric Structures and their Applications (GNSAGA - INdAM).
The second author is funded by the project ``Attrazione e Mobilità dei Ricercatori'' Italian PON Programme (PON-AIM 2018 num. AIM1878214-2). The second and the third authors are supported by the project ``VALERE: VAnviteLli pEr la RicErca" of the University of Campania ``Luigi Vanvitelli''.

\bibliographystyle{acm}
\bibliography{biblio.bib}

\begin{thebibliography}{10}

\bibitem{MR2648536}
{\sc Aubry, Y., McGuire, G., and Rodier, F.}
\newblock A few more functions that are not {APN} infinitely often.
\newblock In {\em Finite fields: theory and applications}, vol.~518 of {\em
  Contemp. Math.} Amer. Math. Soc., Providence, RI, 2010, pp.~23--31.

\bibitem{Bartoli:2020aa4}
{\sc Bartoli, D.}
\newblock Hasse-weil type theorems and relevant classes of polynomial
  functions.
\newblock {\em London Mathematical Society Lecture Note Series, Proceedings of
  28th British Combinatorial Conference, Cambridge University Press\/} (2021),
  43--102.

\bibitem{bartoli2020conjecture}
{\sc Bartoli, D., Csajb{\'o}k, B., and Montanucci, M.}
\newblock On a conjecture about maximum scattered subspaces of
  $\mathbb{F}_{q^6}\times \mathbb{F}_{q^6}$.
\newblock {\em arXiv preprint arXiv:2004.13101\/} (2020).

\bibitem{Bartoli:2021}
{\sc Bartoli, D., Giulietti, M., and Zini, G.}
\newblock The classification of exceptional scattered polynomials of odd
  degree.
\newblock {\em in preparation\/} (2021).

\bibitem{bartoli2020r}
{\sc Bartoli, D., Micheli, G., Zini, G., and Zullo, F.}
\newblock $ r $-fat linearized polynomials over finite fields.
\newblock {\em arXiv preprint arXiv:2012.15357\/} (2020).

\bibitem{MR4190573}
{\sc Bartoli, D., and Montanucci, M.}
\newblock On the classification of exceptional scattered polynomials.
\newblock {\em J. Combin. Theory Ser. A 179\/} (2021), 105386, 28.

\bibitem{MR4173668}
{\sc Bartoli, D., Zanella, C., and Zullo, F.}
\newblock A new family of maximum scattered linear sets in {${\rm PG}(1,
  q^6)$}.
\newblock {\em Ars Math. Contemp. 19}, 1 (2020), 125--145.

\bibitem{MR3812212}
{\sc Bartoli, D., and Zhou, Y.}
\newblock Exceptional scattered polynomials.
\newblock {\em J. Algebra 509\/} (2018), 507--534.

\bibitem{MR4110235}
{\sc Bartoli, D., and Zhou, Y.}
\newblock Asymptotics of {M}oore exponent sets.
\newblock {\em J. Combin. Theory Ser. A 175\/} (2020), 105281, 18.

\bibitem{blokhuis2000scattered}
{\sc Blokhuis, A., and Lavrauw, M.}
\newblock Scattered spaces with respect to a spread in $\mathrm{PG}(n, q)$.
\newblock {\em Geom. Dedicata 81}, 1 (2000), 231--243.

\bibitem{csajbok2018new}
{\sc Csajb{\'o}k, B., Marino, G., Polverino, O., and Zanella, C.}
\newblock A new family of {MRD}-codes.
\newblock {\em Linear Algebra Appl. 548\/} (2018), 203--220.

\bibitem{csajbok2020mrd}
{\sc Csajbok, B., Marino, G., Polverino, O., and Zhou, Y.}
\newblock {MRD} codes with maximum idealizers.
\newblock {\em Discrete Math. 343}, 9 (2020), 111985.

\bibitem{csajbok2018new2}
{\sc Csajb{\'o}k, B., Marino, G., and Zullo, F.}
\newblock New maximum scattered linear sets of the projective line.
\newblock {\em Finite Fields Appl. 54\/} (2018), 133--150.

\bibitem{delsarte1978bilinear}
{\sc Delsarte, P.}
\newblock Bilinear forms over a finite field, with applications to coding
  theory.
\newblock {\em J. Combin. Theory Ser. A 25}, 3 (1978), 226--241.

\bibitem{MR4163074}
{\sc Ferraguti, A., and Micheli, G.}
\newblock Exceptional scatteredness in prime degree.
\newblock {\em J. Algebra 565\/} (2021), 691--701.

\bibitem{MR1042981}
{\sc Fulton, W.}
\newblock {\em Algebraic curves}.
\newblock Advanced Book Classics. Addison-Wesley Publishing Company, Advanced
  Book Program, Redwood City, CA, 1989.
\newblock An introduction to algebraic geometry, Notes written with the
  collaboration of Richard Weiss, Reprint of 1969 original.

\bibitem{gabidulin1985theory}
{\sc Gabidulin, E.~M.}
\newblock Theory of codes with maximum rank distance.
\newblock {\em Problemy Peredachi Informatsii 21}, 1 (1985), 3--16.

\bibitem{MR0463157}
{\sc Hartshorne, R.}
\newblock {\em Algebraic geometry}.
\newblock Springer-Verlag, New York-Heidelberg, 1977.
\newblock Graduate Texts in Mathematics, No. 52.

\bibitem{MR1359909}
{\sc Janwa, H., McGuire, G.~M., and Wilson, R.~M.}
\newblock Double-error-correcting cyclic codes and absolutely irreducible
  polynomials over {${\rm GF}(2)$}.
\newblock {\em J. Algebra 178}, 2 (1995), 665--676.

\bibitem{kshevetskiy2005new}
{\sc Kshevetskiy, A., and Gabidulin, E.}
\newblock The new construction of rank codes.
\newblock In {\em Proceedings. International Symposium on Information Theory,
  2005. ISIT 2005.\/} (2005), IEEE, pp.~2105--2108.

\bibitem{MR3326175}
{\sc Leducq, E.}
\newblock Functions which are {PN} on infinitely many extensions of
  {$\mathbb{F}_p$}, {$p$} odd.
\newblock {\em Des. Codes Cryptogr. 75}, 2 (2015), 281--299.

\bibitem{liebhold2016automorphism}
{\sc Liebhold, D., and Nebe, G.}
\newblock Automorphism groups of {G}abidulin-like codes.
\newblock {\em Arch. Math. 107}, 4 (2016), 355--366.

\bibitem{MR3678916}
{\sc Loidreau, P.}
\newblock A new rank metric codes based encryption scheme.
\newblock In {\em Post-quantum cryptography}, vol.~10346 of {\em Lecture Notes
  in Comput. Sci.} Springer, Cham, 2017, pp.~3--17.

\bibitem{longobardi2021large}
{\sc Longobardi, G., Marino, G., Trombetti, R., and Zhou, Y.}
\newblock A large family of maximum scattered linear sets of {PG}$(1, q^n)$ and
  their associated {MRD} codes.
\newblock {\em arXiv preprint arXiv:2102.08287\/} (2021).

\bibitem{longobardi2021linear}
{\sc Longobardi, G., and Zanella, C.}
\newblock Linear sets and {MRD}-codes arising from a class of scattered
  linearized polynomials.
\newblock {\em J. Algebraic Combin.\/} (2021), 1--23.

\bibitem{lunardon2000blocking}
{\sc Lunardon, G., and Polverino, O.}
\newblock Blocking sets of size $q^t+ q^{t- 1}+ 1$.
\newblock {\em J. Combin. Theory Ser. A 90}, 1 (2000), 148--158.

\bibitem{lunardon2018nuclei}
{\sc Lunardon, G., Trombetti, R., and Zhou, Y.}
\newblock On kernels and nuclei of rank metric codes.
\newblock {\em J. Algebraic Combin. 46\/} (2017), 313--340.

\bibitem{lunardon2018generalized}
{\sc Lunardon, G., Trombetti, R., and Zhou, Y.}
\newblock Generalized twisted gabidulin codes.
\newblock {\em J. Combin. Theory Ser. A 159\/} (2018), 79--106.

\bibitem{marino2020mrd}
{\sc Marino, G., Montanucci, M., and Zullo, F.}
\newblock {MRD}-codes arising from the trinomial $x^q+ x^{q^3}+ cx^{q^5}\in
  \mathbb{F}_{q^6} [x]$.
\newblock {\em Linear Algebra Appl. 591\/} (2020), 99--114.

\bibitem{MR3087321}
{\sc Mullen, G.~L., and Panario, D.}
\newblock {\em Handbook of finite fields}.
\newblock CRC Press, 2013.

\bibitem{NPZ}
{\sc Neri, A., Santonastaso, P., and Zullo, F.}
\newblock Extending two families of maximum rank distance codes.
\newblock {\em arXiv preprint arXiv:2104.07602\/} (2021).

\bibitem{polverino2020number}
{\sc Polverino, O., and Zullo, F.}
\newblock On the number of roots of some linearized polynomials.
\newblock {\em Linear Algebra Appl. 601\/} (2020), 189--218.

\bibitem{ravagnani2016rank}
{\sc Ravagnani, A.}
\newblock Rank-metric codes and their duality theory.
\newblock {\em Designs, Codes and Cryptography 80}, 1 (2016), 197--216.

\bibitem{MR3239294}
{\sc Schmidt, K.-U., and Zhou, Y.}
\newblock Planar functions over fields of characteristic two.
\newblock {\em J. Algebraic Combin. 40}, 2 (2014), 503--526.

\bibitem{MR3543528}
{\sc Sheekey, J.}
\newblock A new family of linear maximum rank distance codes.
\newblock {\em Adv. Math. Commun. 10}, 3 (2016), 475--488.

\bibitem{sheekeysurvey}
{\sc Sheekey, J.}
\newblock {MRD} codes: constructions and connections.
\newblock {\em Combinatorics and Finite Fields: Difference Sets, Polynomials,
  Pseudorandomness and Applications 23\/} (2019).

\bibitem{MR2450762}
{\sc Silva, D., Kschischang, F.~R., and K\"{o}tter, R.}
\newblock A rank-metric approach to error control in random network coding.
\newblock {\em IEEE Trans. Inform. Theory 54}, 9 (2008), 3951--3967.

\bibitem{SLAVOV201760}
{\sc Slavov, K.}
\newblock An application of random plane slicing to counting
  $\mathbb{F}_q$-points on hypersurfaces.
\newblock {\em Finite Fields and Their Applications 48\/} (2017), 60--67.

\bibitem{zanella2019condition}
{\sc Zanella, C.}
\newblock A condition for scattered linearized polynomials involving {D}ickson
  matrices.
\newblock {\em J. Geom. 110}, 3 (2019), 1--9.

\bibitem{zanella2020vertex}
{\sc Zanella, C., and Zullo, F.}
\newblock Vertex properties of maximum scattered linear sets of {PG}$(1,q^n)$.
\newblock {\em Discrete Math. 343}, 5 (2020), 111800.

\end{thebibliography}
\end{document}